%% file: wwv.tex
\documentclass[submission,creativecommons]{eptcs}
\providecommand{\event}{WWV 2011}
\usepackage{amsmath,amsthm,amssymb}
\usepackage{xspace,listings,afterpage}

\title{A type checking algorithm for qualified session types} 

\author{Marco Giunti  
\institute{INRIA and LIX, \'Ecole Polytechnique, France}}
\def\titlerunning{A type checking algorithm for qualified session types}
\def\authorrunning{M.~Giunti}

\input{macros}  
\newcommand{\safe}{{\sf safe}}
\newcommand{\Un}{{\sf un}}   
\newcommand{\used}{{\sf map}}
\newtheorem{theorem}{Theorem}[section]
\newtheorem{corollary}[theorem]{Corollary}
\newtheorem{lemma}[theorem]{Lemma}

\theoremstyle{remark}
\newtheorem*{remark}{Remark}

\newtheorem*{claim}{Claim}

\begin{document}
\maketitle

\begin{abstract}
We present a type checking algorithm for establishing a session-based discipline
in the pi calculus of Milner, Parrow and Walker. Our session types are qualified
as linear or unrestricted. Linearly typed communication channels are guaranteed
to occur in exactly one thread, possibly multiple times; afterwards they evolve
as unrestricted channels. Session protocols are described by a type constructor
that denotes the two ends of one and the same communication channel. We ensure the
soundness of the algorithm by showing that processes consuming all linear
resources are accepted by a typing system preserving typings during the
computation and that type checking is consistent w.r.t. structural congruence. 
\end{abstract}

\section{Introduction}
\label{sec:intro}

Session types allow a concise description of protocols by detailing
the sequence of messages involved in each particular run of the
protocol. Introduced for a dialect of the pi
calculus~\cite{honda.vasconcelos.kubo:language-primitives,THK}, the
concept has been transferred to different realms, including functional
and object-oriented programming and operating systems; refer
to~\cite{dezani-de-liguoro:session-types-overview} for a recent
overview.

To illustrate, consider the problem of designing  a 
web system for the scheduling of meetings. 
In our example, the  system is implemented by means of a web service 
repeatedly waiting for requests to create a poll. 
Once invoked, the service instantiates a fresh session for the poll and launches
a thread for managing it. 
In the pi calculus~\cite{MPW92} the session could be modeled as a communication
channel for the exchange of the messages required by the scheduling protocol.  
The fresh  channel for the poll is forwarded back to the invoker on the channel
she has provided in order to receive the information needed for the start
of the poll: the title and a tentative date for the meeting.  
Afterwards the thread repeatedly waits for possible date
proposals from the participants of the poll.  
 \begin{align*}
P_1 =!x(y).\NR{p}(\SEND yp p({\sf title}).p({\sf date}).!p({\sf date})) 
\end{align*}
In order to have some guarantee on the behavior of the executable system,  a
static analysis of its code should be performed during the
compilation. A typed analysis  permits indeed to verify the desired
properties of the protocol, namely that there is exactly one title and at least
one date proposal for the meeting. To this aim we need to enforce that the
capability forwarded to the caller consists in (i) send  a string for
the title and afterwards (ii) send one or more dates. 
This behavior could be described by relying on  polymorphic types qualified as
linear or unrestricted.
The idea is to introduce qualifiers for  types describing a session 
and to allow a linear usage of a session to evolve to an unrestricted usage.
This approach has been indeed advocated as effective
independently from any programming language~\cite{vvBeatcs}.
A qualified session type for the poll channel sent to the invoker is the one
below.  
\begin{align*}
S_2 = \lin\,\OUT {\sf string}\lin\,\OUT{\sf date}S_4 
\qquad
S_4 =\un\,\OUT{\sf date}S_4
\end{align*}
The session type first describes the
sending of a string to set the title of the meeting; 
such usage is qualified as linear because a title for the schedule
is required. Similarly, the continuation  type for sending the date of the 
schedule is qualified as linear because a date has to be set in order to start
the poll. Lastly zero or more date proposals could be send on the poll
channel; this behavior is described by a unrestricted recursive type.
The continuation of the service $P_1$ is described by the type $S_1$ below
that could be seen as the ``dual'' of $S_2$. 
\vspace{-1.1em}
\begin{align*}
S_1 =\lin\,\IN {\sf string}\lin\,\IN{\sf date}{S_3} 
\qquad
S_3 =\un\,\IN{\sf date}S_3.  
\end{align*}
The session type describes the behavior of receiving the title and one 
or more date proposals for the schedule. The  
receiving of the title and of the date proposal  are both qualified as linear
because this information is mandatory. Eventually, zero or more date
proposals  will arrive afterwards. The unbounded behavior of receiving such
proposals is described by the unrestricted recursive type  $S_3$.
The usage of the poll channel is described by a  type constructor 
$\cotype{S_1}{S_2}$ representing the concurrent behavior of the two channel
ends~\cite{GiuntiV10}. The intuition is that in typing (the continuation of)
service $P_1$ the type  $\cotype{S_1}{S_2}$ is split into two parts: the
linear output end point is used to type the delegation of one end of the session
to the invoker while the linear input end point is used to type the
continuation process.   

While the idea of split types and contexts is clear and concise, the
inherent non-determinism contained in its formulation makes a direct
implementation infeasible. 
Algorithmic solutions for linear functional languages   avoid to split the
context into parts before checking a complex expression by passing the entire
context as input to the first subexpression and have it return the unused
portion as an output~\cite{walker:substructural-type-systems}.
In the setting of concurrent computations, the  idea  is
that when typing a parallel process $P\PAR Q$ the set of linear identifiers used
by $P$ must be calculated in order to remove it before  type checking~$Q$.  
This approach, previously  outlined for linear types of pi
calculus in~\cite{kpt:linearity}, has been implemented in the
session system of~\cite{GaySJ:substp} by representing each channel end with
a distinct identifier.  
 
In this paper, we propose an algorithm to  check protocols described
by types of the form $\cotype{S_1}{S_2}$ where each $S_i$ is a qualified session
type depicting one end of the communication.
Channels could evolve from linear to unrestricted usage. 
Reasoning at the type level, we do implement split by forbidding the utilization of
used parts of types and by a careful analysis of qualifiers.
This construction permits us to show that
(i) type checked processes  are accepted by a typing system satisfying subject
reduction and that (ii) type checking preserves structural congruence.    

More in detail, type checking  relies on the definition of several unambiguous
patterns. The patterns for linear input and  output processes do
return a {\it marked} context. In the body of the function a recursive call to
type check the continuation is launched. If an exception is not raised, this
call returns in output a context. First, to ensure a subsequent linear usage to
be finished within the continuation we verify the type for the
variable in the context to be unrestricted. Second, to prohibit the use of the
variable in the next thread we return a context with an  ``unusable'' mark for
the type of the variable. Similarly, in delegating a channel end of a session we
pass to the checking function for the continuation a context with an unusable
mark for the delegated type. Under replication, we do no admit to return
new typings marked as unusable, which would imply consumption of a linear resource.
Lastly, the algorithm succeeds if the context returned by the top-level call of
the type checking function does not contain linear types.  

The remainder of the paper is as follows.
In Section \ref{sec:pi-calculus} we introduce session types and pi calculus.
Section \ref{sec:algorithm} presents the type checking algorithm.
Section \ref{sec:soundness} is devoted to establish the soundness of our approach.
In the last part of the section we  investigate the expressiveness of the
algorithm.
Some examples of the concrete execution of the algorithm are illustrated in
Section \ref{sec:examples}.
We conclude in Section \ref{sec:discussion} by discussing limitations and
future work.
  
\section{Pi calculus} 
\label{sec:pi-calculus}
\input{fig-syntax}

This section introduces the syntax and the semantics of the typed pi calculus.
The definition is in \mbox{Figure~\ref{fig:syntax}}. 
We consider channel types of the form
$\cotype{S}{S}$ where $S$ is a  type  describing the behavior
of a channel end point.
An end point type $S$ can be a pre type qualified with \lin or \un, a
recursive type or a type variable.
Each qualifier in a type controls
the number of times the channel can be used at that point: exactly
once for \lin; zero or more times for \un.
A pre type of the form $\sendt$ describes a channel end able to send a
variable of type $T$ and to proceed as prescribed by $S$. Similarly, pre
type $\recvt$ describes a channel end able to receive a variable of type
$T$ and continue as~$S$.
Pre type $\End$ describes a channel end on which no further
interaction is possible.
For recursive (end point) types we rely on a set of type variables,
ranged over by~$a$. Recursive types are required to be contractive,
that is, containing no subexpression of the form $\mu a_1\dots\mu
a_n.a_1$.
%
%
\emph{Type equality} is not syntactic. Instead, we define it as the
equality of regular infinite trees obtained by the infinite unfolding
of recursive types, \emph{modulo pair commutation}. 
The formal definition, which we omit, is co-inductive. 
In this way we use types
$\cotype{\mu a.\lin\OUT{\un\,\End}\lin\IN{\un\,\End} a}{{\un\,\End}}$ and
$\cotype
{\un\,\End}{\lin\OUT{\un\,\End}\mu b.\lin\IN{\un\,\End}\lin\OUT{\un\,\End} b}$
interchangeably, in any mathematical context.
This allows us never to consider a type $\rect$ explicitly (or $a$ for
that matter). Instead, we pick another type in the same equivalence
class, namely $S\subs\rect a$. If the result of the process turns out
to start with a $\mu$, we repeat the procedure. Unfolding is bound to
terminate due to contractiveness.
In other words, we take an equi-recursive view of
types~\cite{PierceBC:typpl}.
 
%
%
The syntax and the semantics of pi calculus processes are those
of~\cite{MPW92} but for restriction, for
which we require type annotation.
This is only to facilitate type checking and has no impact on the semantics.
We rely on a set of variables, ranged over by $x,y,z$.  
For processes we have (synchronous, unary) output and input, in the forms 
$\SEND xyP$ and $\RECEIVE xyP$,
as well as a parallel composition, annotated scope restriction, replication and
the terminated process. 
%
%
The binders for the language appear in parenthesis: $x$ is bound in
both $\RECEIVE yxP$ and $\NR {x:T}P$. Free and bound variables in
processes are defined accordingly, and so is alpha conversion,
substitution of a variable $x$ by a variable $z$ in a process~$P$,
denoted $P\subs zx$. We follow Barendregt's variable convention,
requiring bound variables to be distinct from each other and from free variables  
in any mathematical context.

Structural congruence is the smallest relation on processes including
the rules in Figure~\ref{fig:syntax}. The first three rules say that parallel
composition is commutative, associative and has $\INACT$ as neutral
element. The last rule on the first line captures the essence of
replication as an unbounded number of identical processes.  The rules
in the second and third line deal with scope restriction. The first, scope
extrusion, allows the scope of $x$ to encompass $Q$; due to variable
convention, $x$ bound in $\NR {x:T}P$, cannot be free in $Q$.  
The next rule allows exchanging the order
of restrictions.
The rules on the third line state that restricting over a terminated process has
no effect.  Since it makes poor sense to declare a new variable with a linear type
for a terminated process, we require the type annotation to be unrestricted. 
The reduction  is the smallest relation on processes
including the rules in Figure~\ref{fig:syntax}. The \rcom rule
communicates a variable $z$ from an output prefixed one $\SEND xzP$ to an
input prefixed process $\RECEIVE xyQ$; the result is the parallel
composition of the continuation processes, where the bound variable
$y$ is replaced by the variable $z$ in the input process.  
The rules on the last line allow
reduction to happen underneath scope restriction and parallel
composition, and incorporate structural congruence into reduction.

\section{Type checking algorithm}
\label{sec:algorithm}
In this section we present an algorithm for type checking a pi calculus
process given a typing context. 
Type checking  relies  on the definition of several
patterns which, for the sake of clarity, we  present in a declarative style.
Lastly, in Figure~\ref{fig:algorithm} we present an excerpt of the {\it ML}
implementation. 

\medskip\noindent
{\bf Contexts}.
 We let $\Gamma$ be a map from variables to  types and 
the {\it void} symbol, noted $\circ$;  a void symbol  permits to mark an end point 
as unusable.
\begin{align*}
M,N,O &\grmeq S \PAR \circ   && \text{entry}
\\
\Gamma &\grmeq \emptyset \PAR \Gamma,x:M \PAR \Gamma,x:\cotype{M}{N} &&\text{context}
\end{align*} 
\noindent
Context updating, noted $\uplus$,  is the procedure effected by the typing
system to  transform a void entry in an end point entry:
$
\Gamma,x:\circ\, \uplus x:M =\Gamma,x:M 
$. 
A {\it safe} context is a map from variable to safe entries; 
we let the predicate $\safe(\Gamma)$  hold whenever $x\in\dom(\Gamma)$ implies
$\safe(\Gamma(x))$.
A linear channel type  is safe if (i) the type of the
variable sent
in output corresponds to the type expected in input and (ii) the expected type
for the input is safe and (iii) the continuation is safe. 
For an unrestricted channel type we require (i) and (ii):  (iii) will be enforced by
the type system. 
\begin{align*}
&\safe(M)
\\
&\safe(\cotype{M_1}{M_2}) \hspace{2cm} 
\exists i\in\{1,2\}.\ M_i=\circ,\un\,\End 
\\
&\safe(\cotype{\lin\, \IN TS_1}{\lin\,\OUT TS_2})  =
\safe(T)\land \safe(\cotype{S_1}{S_2})
\\
&\safe(\cotype{\un\, \IN TS_1}{\un\,\OUT TS_2})  =
\safe(T) 
\end{align*}
\noindent
A context is {\it unrestricted} if it contains only
unrestricted or void entries.  We let $\Un(\Gamma)$ whenever $x\in\dom(\Gamma)$
implies $\Un(\Gamma(x))$.
\begin{align*} 
\hspace{0.7cm}&\Un(\circ) 
\\ 
&\Un( \un\,p )
\\
&\Un(\cotype{M}{N})= \Un(M)\,\land\,\Un(N)   
\end{align*}

\medskip\noindent
{\bf Patterns}.
We present typing rules for processes of the form $\Gamma_1\vdash
P\ret\Gamma_2$ where   $\Gamma_1$ is a context received in input
and $\Gamma_2$ is a context produced as output.
Given that $\Gamma_1$ is a context such that $\safe(\Gamma_1)$, the
rules are chosen deterministically by inspecting (i) the shape of the
context and (ii) the shape of the process, in the following way.  
Each rule is implemented as a pattern of a function with
signature $\sf check (g:context,p:process):context$. 
For each function call with a safe context parameter, zero or
one pattern does match; in the first case
a pattern exception indicating the reject of the process is raised  while
in the second case a context is returned in output to the caller. 
The rules for variables have the form
$\Gamma_1\vdash v:T\ret\Gamma_2$ and are implemented as patterns of a function
with signature {\sf checkVar (g:context, v:var):context }.
In the rules below the output context is
obtained by setting to void the linear assumptions used to type the variable. 
The last three rules permit to resolve any ambiguity in typing an unrestricted end
point type with an unrestricted  channel type.
\begin{gather*}  
    \tag*{\avarl,\avaru}
    \frac{
    \Gamma=\Gamma_1,x\colon \lin\,p,\Gamma_2
    }{
    \Gamma\vdash x \colon  \lin\,  p 
    	\ret\Gamma_1,x \colon \circ,\Gamma_2 
    }
    \qquad
    \frac{
    \Gamma=\Gamma_1,x\colon  \un\,p ,\Gamma_2
    }
    {
    \Gamma\vdash x \colon  \un\,p  \ret\Gamma
    }
    \\[2mm] 
    \tag*{\avarlll}
    \frac{\Gamma=
    \Gamma_1,x\colon \cotype{ \lin\,p_1 }{ \lin\,p_2 },\Gamma_2
    }
    {
    \Gamma\vdash 
    x \colon \cotype{ \lin\,p_1 }{ \lin\,p_2} 
    \ret\Gamma_1,x\colon\cotype{\circ}{\circ},\Gamma_2
    }
    \\
    \tag*{\avarllr}
    \frac{\Gamma=
    \Gamma_1,x\colon \cotype{ \lin\,p_1 }{ \lin\,p_2 },\Gamma_2
    }
    {
    \Gamma\vdash 
    x \colon \cotype{ \lin\,p_2 }{ \lin\,p_1} 
    \ret\Gamma_1,x\colon\cotype{\circ}{\circ},\Gamma_2
    }
    \\[2mm] 
    \tag*{\avarlsl,\avarlsr}
    \frac{
    \Gamma=\Gamma_1,x\colon\cotype{\lin\,p  }{N},\Gamma_2
    }{
     \Gamma\vdash x \colon \lin\, p
    	\ret\Gamma_1,x\colon \cotype{\circ}{N},\Gamma_2
    }
    \qquad 
    \frac{
    \Gamma=\Gamma_1,x\colon\cotype{M}{ \lin\,p },\Gamma_2
    }{
     \Gamma\vdash x \colon \lin\, p 
    	\ret\Gamma_1,x\colon \cotype{M}{\circ},\Gamma_2
    } 
    \\[2mm] 
    \tag*{\avaruul,\avaruur}
    \frac{\Gamma=\Gamma_1,x\colon\cotype{ \un\,p_1 }{ \un\,p_2 },\Gamma_2
    }
    {\Gamma\vdash x \colon \cotype{\un\,p_1}{\un\,p_2} \ret\Gamma 
    }
    \qquad
    \frac{\Gamma=\Gamma_1,x\colon\cotype{ \un\,p_1 }{ \un\,p_2 },\Gamma_2
    }
    {\Gamma\vdash x \colon \cotype{\un\,p_2}{\un\,p_1} \ret\Gamma 
    }
    \\[2mm]   
    \tag*{\avarusr,\avarusl}
    \frac{
    \Gamma=\Gamma_1,x\colon \cotype{ \un\,p }{N},\Gamma_2 \quad \un\,p\ne N
    }
    {
    \Gamma\vdash x \colon \un\, p \ret\Gamma
    }
    \qquad
            \frac{ \Gamma=\Gamma_1,x\colon \cotype{M}{ \un\,p },\Gamma_2  \quad
\un\,p\ne M
    }
    {
\Gamma\vdash x \colon \un\, p \ret\Gamma
    	}
    \\[2mm]
    \tag*{\avaruend}
    \frac{ \Gamma=\Gamma_1,x\colon \cotype{\un\,\End}{\un\,\End},\Gamma_2
     }
    {
\Gamma\vdash x \colon \un\, \End \ret\Gamma
} 
\end{gather*}

\noindent
Rule \aoutl is to type processes sending variables on a channel used in
linear mode given that the type for the channel in the context is an   end point.  
The context changed by setting the channel to void is used to check the  sent
variable at the expected type and in turn to return a new context. 
The new context updated with the continuation type for the linear
channel is passed as parameter in the call for checking the continuation
process. To  ensure a linear use of the channel to be finished within the
continuation, we verify that the 
context returned by the call for the continuation does contain an unrestricted
typing for the channel. Finally, the returned context is
given as output with the typing for the channel set to void.

\begin{gather*}
    \tag*{\aoutl}\frac{
      \Gamma_1,x:\circ\,\vdash y \colon T\ret\Gamma_2 
      \qquad
      \Gamma_2\uplus  x\colon S \vdash P\ret\Gamma_3,x:M 
      \qquad \Un(M)  
    }{
      \Gamma_1,x\colon  \lin\,\OUT TS \vdash \sendp\ret\Gamma_3,x:\circ\  
    }
\end{gather*}

\noindent
Rules \aoutcll,\aoutclr are used when the entry for the linear output
in the context is a channel type. The rules are implemented by the pattern
\aoutl. In returning the context we set one end of the channel type to void
while we leave the other end as it has been received in input. 
 
\begin{gather*}     
   \tag*{\aoutcll}
    \frac{
      \Gamma_1,x\colon { \lin\,\OUT TS} 
      \vdash \sendp\ret\Gamma_3,x: {\circ}  
     }{
      \Gamma_1,x\colon \cotype{ \lin\,\OUT TS}{N}
      \vdash \sendp\ret\Gamma_3,x:\cotype{\circ}{N}
    }
    \\[2mm]
       \tag*{\aoutclr}
    \frac{
       \Gamma_1,x\colon { \lin\,\OUT TS} 
      \vdash \sendp\ret\Gamma_3,x: {\circ}  
     }{
      \Gamma_1,x\colon \cotype{M}{ \lin\,\OUT TS}
      \vdash \sendp\ret\Gamma_3,x:\cotype{M}{\circ}
    } 
\end{gather*}

\noindent
For sending a variable on an unrestricted channel we require the sent
variable to be typable by the same context received in input; that is, the type
for the unrestricted output channel must be recursive. 
The context obtained by the typing for the variable is then used to
call the checking function for the continuation process.   
 
\begin{gather*}
 \tag*{\aoutu}
    \frac{
      \Gamma_1,x: S \vdash  v \colon T\ret\Gamma_2 
      \qquad
      \Gamma_2\vdash P\ret\Gamma_3 
      \qquad
      S=\un\,\OUT TS 
    }{
      \Gamma_1,x\colon  \un\,\OUT TS \vdash \sendp\ret\Gamma_3
    }
    \qquad  
 \\[2mm]
 \tag*{\aoutclu}
    \frac{
      \Gamma_1,x\colon\cotype{ S }{N}\vdash v \colon
T\ret\Gamma_2 
      \qquad
      \Gamma_2\vdash P\ret\Gamma_3
      \qquad
      S= \un\,\OUT TS 
    }{
      \Gamma_1,x\colon \cotype{ \un\,\OUT TS }{N}\vdash
\sendp\ret\Gamma_3
    }
    \\[2mm]
\tag*{\aoutcru}    
    \frac{
      \Gamma_1,x\colon\cotype{M}{ S }\vdash v \colon
T\ret\Gamma_2 
      \qquad
      \Gamma_2\vdash P\ret\Gamma_3 
      \qquad 
      S=\un\,\OUT TS
    }{
      \Gamma_1,x\colon \cotype{M}{ \un\,\OUT TS }\vdash
\sendp\ret\Gamma_3
    }
\end{gather*}
 
\noindent
To type a linear usage of an input we require the expected type to
agree with the type of the input channel and the continuation type for the channel to
 be consumed within the continuation. This
is implemented by requiring that the context returned by the call for the
continuation does map the variable to an unrestricted type. 
We also require a linear usage for the variable bound by the input to be
finished within its local scope. Lastly, the call returns a context (i) with the
type for the linear variable set to void and (ii) pruned by the variable bound
by the input prefix. This is the rationale of rule  \ainl  
 and or rules \aincll,\ainclr which are used for variables having
respectively an end point or a channel type. 

\begin{gather*}    
\tag*{\ainl}
    \frac{
      \Gamma_1, x\colon S, y\colon T\vdash P\ret\Gamma_2,x:M,y:O
      \qquad \Un(M)  
      \qquad \Un(O)
      }{
      \Gamma_1,x\colon \lin\IN TS\vdash \receivep\ret\Gamma_2,x:\circ
    }
     \\[2mm]
\tag*{\aincll}
    \frac{
      \Gamma_1,x\colon {\lin\IN TS}
      \vdash \receivep\ret\Gamma_2,x: {\circ} 
     }{
      \Gamma_1,x\colon \cotype{\lin\IN TS}{N}
      \vdash \receivep\ret\Gamma_2,x:\cotype{\circ}{N} 
    }
    \\[2mm]
\tag*{\ainclr}    
    \frac{
          \Gamma_1,x\colon {\lin\IN TS}
      \vdash \receivep\ret\Gamma_2,x: {\circ} 
     }{
      \Gamma_1,x\colon \cotype{M}{\lin\IN TS}
      \vdash \receivep\ret\Gamma_2,x:\cotype{M}{\circ} 
    }
\end{gather*}

\noindent
The rules for unrestricted input take the context
received in input and add the bound variable at the expected type in order to
type the continuation. The context returned by the call of the
checking function for the continuation needs to be first verified to ensure that
the type for the bound variable is unrestricted, and then pruned by the
variable to be returned in output.

\begin{gather*}  
 \tag*{\ainu}
    \frac{
      \Gamma_1,x\colon S,y\colon T \vdash P\ret\Gamma_2,y:O   
      \qquad\Un(O)
      \qquad S=\un\IN TS}
      {
      \Gamma_1,x\colon  \un\IN TS \vdash  \receivep\ret\Gamma_2
      }
    \\[2mm]
\tag*{\ainclu}
     \frac{ 
      \Gamma_1,x\colon\cotype{S}{N},y\colon T
      \vdash P\ret\Gamma_2,y:O 
      \qquad \Un(O)
      \qquad S=\un\IN TS
      }{
      \Gamma_1,x\colon\cotype{ \un\IN TS }{N}\vdash \receivep\ret\Gamma_2
    }
\\[2mm]
\tag*{\aincru}
   \frac{   
     \Gamma_1,x\colon\cotype{M}{S},y\colon T
                 \vdash P\ret\Gamma_2,y\colon O  
     \qquad \Un(O)
     \qquad S=\un\IN TS}
     {
      \Gamma_1,x\colon\cotype{M}{\un\IN TS}\vdash \receivep\ret\Gamma_2 
     }
\end{gather*}

\noindent
To type an inert process by using \ainact any context suffices; the
context received in input is forwarded in output. To type
a parallel process in \apar we check the first thread with the  context received
in input. This operation returns in output a  context that is used to type-check
the next thread. The context returned by the last typing is forwarded in
output. While imposing an order on parallel processes could appear  restrictive, in
Section~\ref{sec:soundness} we will show that the chosen order makes no difference. 
 
\begin{gather*}
 \tag*{\ainact,\apar}
     \Gamma \vdash  \INACT\ret\Gamma
 \qquad
 \frac{
    \Gamma_1\vdash P\ret\Gamma_2
    \qquad
    \Gamma_2\vdash Q\ret\Gamma_3
    }{
    \Gamma_1\vdash P \PAR Q \ret\Gamma_3
    }
 \end{gather*}

\noindent
In order to type a process generating a new channel, in rule \ares we
require
the typing for the channel to be safe; if it is not, the algorithm stops and an
  exception is raised. Similarly to the input cases, if a linear usage is
prescript for the new variable then it must be finished within its scope. 

\begin{gather*}
\tag*{\ares}
    \frac{
      \safe (T)
      \quad
      \Gamma_1,y\colon T \vdash  P\ret\Gamma_2,y:O 
      \quad 
      \Un(O)
     }{
       \Gamma_1 \vdash  \NR{y\colon T}P\ret\Gamma_2
    }
    \qquad 
\end{gather*}

\noindent
The rule for replication \arepl is below.  
In the call for checking the process under the replication 
we require the context returned in output to be equal to the one received in
input. Indeed, a change in the output context would be obtained by introducing
a void symbol  indicating that  a linear resource has been consumed. 
This must clearly be forbidden under replication. 
On contrast, we allow to return  linear entries in order to type check the next
thread.

\begin{gather*}
 \tag*{\arepl}
    \frac{
        \Gamma_1 \vdash  P\ret\Gamma_2
        \quad
        \Gamma_2=\Gamma_1  
    }{
      \Gamma_1 \vdash  !P\ret\Gamma_2
    }  
\end{gather*} 

\begin{lemma}
\label{lem:monotonic}
If $\safe(\Gamma_1)$ and $\Gamma_1\vdash P\ret\Gamma_2$ then
$\dom(\Gamma_2)=\dom(\Gamma_1)$ and $\safe(\Gamma_2)$. 
\end{lemma}

\input{ml-algorithm} 
\medskip\noindent
{\bf Type checking}.
Having defined typing rules corresponding to patterns of the
checking function, we devise an
algorithm for  establish a session-based type discipline.
Figure~\ref{fig:algorithm} presents the   ML  definition for types, processes and
the type checking function. Type {\sf context} associates variables
to entries which are formed apart 
the end point and the channel type. 
The function {\sf safe}  returns in output the same context received in input
whenever the context satisfies the safe predicate, otherwise
it generates an exception. 
Function {\sf unVar} takes as  parameters a context and a
variable and verifies that the type for the variable in the context 
is unrestricted; in this case the context is returned in output, otherwise 
an exception is raised. 
Functions {\sf remove} and {\sf setVoid} do
perform the required operations and return the updated context. 
We also need auxiliary functions to push and pop entries
to and from the context stack; we omit all the details.
  
The {\sf check} function, the kernel of the type checking procedure, is
defined
by the union of the patterns for the rules introduced in the current section.
In order to illustrate the   mechanism, we draw the translation of some
patterns. 
In patterns for variables and in  \ainl  we assume the variable on the top of
the
context $z$ to be equal to the variable $x$ respectively for the value to type and
for the input prefix of the process. 
The {\sf checkVar} function is called in patterns for output in order
to type the sent variable and obtain in output a context to pass together with the
the continuation to the checking function. 
In \ainl we launch the recursive call of the check function by passing as parameters
the updated context and the continuation process. After checking that the type for
both channel 
 $x$ and the variable bound by the input are unrestricted
in the returned context, we return the context with the type for $x$ set to void. 
In the pattern for \ares we launch the {\sf check} function by passing as parameters 
the context with the new entry and the
continuation process.
The inner call of the safe function immediately raises an
exception if the type for the bound variable is not safe.
Lastly,  we first control that in the returned context the variable is unrestricted
 and then we return the context pruned by the variable.  
The algorithm is implemented by the {\sf typeCheck} function.
The function receives in input a context and a process.
If the context received in input is not safe then 
the function exits immediately.
Otherwise, a context is returned in input  
provided that an exception has not been raised. 
The exception could raise 
(i) when no pattern matching is possible for the chosen derivation or
(ii) when a call of the {\sf safe} function in \ares fails or
(iii) when  call of  {\sf unVar} function fails. 
Since the choice of patterns is deterministic for safe contexts, no
backtracking is needed.
Lastly, the process is accepted by the algorithm whenever the returned context  
satisfies the {\sf un} predicate defined in Section~\ref{sec:pi-calculus}. 

\begin{lemma} 
\label{lem:deterministic}
If  $\safe(\Gamma)$ then {\sf check}$(\Gamma,P)$
matches zero or one patterns.
\end{lemma} 
 
\begin{lemma}
If {\sf check}$(\Gamma',P')$ has been recursively invoked by
{\sf typeCheck}$(\Gamma,P)$ then we have $\safe(\Gamma')$. 
\end{lemma}
\begin{proof}
A call is a match of a pattern $\Gamma_1\vdash P\ret \Gamma_2$ which is an axiom
whenever $P=\INACT$, and has been inferred from an hypothesis starting with a type
environment $\Delta $ on the left otherwise.
We proceed by induction 
and show a stronger result, namely that  
$\safe(\Delta )$ implies $\safe(\Gamma_1)$.
We close the proof by applying Lemma~\ref{lem:monotonic}, and eventually
by exploiting transitivity in cases for output and parallel composition.
\end{proof} 

\begin{corollary}
\label{cor:deterministic}
If {\sf check}$(\Gamma',P')$ is a call invoked during the execution of {\sf
typeCheck}$(\Gamma,P)$ then there are zero or one patterns to match.
\end{corollary} 
   
\section{Soundness}
\label{sec:soundness} 
\input{fig-type-system}

This section is devoted to establishing the soundness of the algorithm.
To this aim we  project 
the pattern rules presented in Section~\ref{sec:algorithm} into the typing system  of
Figure~\ref{fig:split-system}, which satisfies subject reduction~\cite{GiuntiV10}. 
The syntax of types and processes occurring
in Figure~\ref{fig:split-system} is that of Figure~\ref{fig:syntax}.
Contexts $\I$ are a map from variables to types $T$:
\[
 \I \grmeq  \emptyset \PAR \I, x\colon T  \ .
\]   
Typing rules in Figure~\ref{fig:split-system} are based on a declarative
definition
of context splitting; the intuition is that
unrestricted types are copied into both contexts, while linear types are
placed in one of the two resulting contexts. 
We refer to~\cite{GiuntiV10} for the details.

We introduce preliminary Lemmas and Definitions which will be useful to prove
the main result of this section.
Given a judgment $\Gamma_1\vdash P\ret\Gamma_2$ of the algorithmic
system of Section~\ref{sec:algorithm} , we let 
the  used closure of a type context $\Gamma_1$ w.r.t. $\Gamma_2$, noted
$\Gamma_1 \ret \Gamma_2$, be the typing context $\emptyset $ whenever
$\Gamma_1=\emptyset$, and be defined by 
$(\Gamma_1\ret\Gamma_2)(x)=\Gamma_1(x)\ret\Gamma_2(x)$ otherwise:
\begin{align*}
\circ \ret \circ  &= \circ
&
\lin\,p_1\ret \lin p_1  &= \circ
\\
 \lin\,p_1\ret \circ &= \lin\,p_1 
&
 \un\,p_1 \ret  \un\,p_1  &=\un\,p_1   
 \\
\cotype{M}{N}\ret \cotype{M'}{N'}&=
\cotype{M\ret M'}{N\ret N'} \ .
\end{align*}

\noindent
The map operation projects a type environment $\Gamma$ 
into a context $\I$ of Figure \ref{fig:split-system}.
When applied to a used closure, it permits to map  
linear typings which do not change from 
$\Gamma_1$ to $\Gamma_2$  into the $\un\,\End$ type.
\begin{align*}
\used(\circ)&=\un\,\End & 
\used(S)&=S 
\\
\used(\cotype{M}{N})&=\cotype{\used(M)}{\used(N)} 
&
\used(\Gamma)&=\bigcup_{x\in\dom(\Gamma)} x:\used(\Gamma(x))
&\end{align*}
 
\begin{lemma}\label{lem:usedefined}
Assume $\safe(\Gamma_1)$. If $\Gamma_1\vdash P\ret\Gamma_2$ then
$\used(\Gamma_1\ret\Gamma_2)$ is defined. 
\end{lemma}

A  used closure generated by the algorithmic system is sufficient to type a process
with the system $\vdash_D$,  as we will show in a
nontrivial manner below.
We need a couple of lemmas for strengthening judgments of the algorithmic
system and weaken judgments of the split-based system. 

\begin{lemma}[Algorithmic strengthening]
\label{lem:alg-strength}
The following hold. 
\begin{enumerate}
\item If $\Gamma_1,x:\lin p\vdash P\ret \Gamma_2,:x:\lin p$ then
$\Gamma_1\vdash P\ret \Gamma_2$;
\item If $\Gamma_1,x:\cotype{\lin p}{S}\vdash P\ret \Gamma_2,
x:\cotype{\lin p}{N}$  then $\Gamma_1,x:S\vdash P\ret \Gamma_2,x:N'$;
\item If $\Gamma_1,x:\cotype{M}{\lin p}\vdash P\ret \Gamma_2,
x:\cotype{M'}{\lin p}$ then $\Gamma_1,x:M\vdash P\ret \Gamma_2,x:M'$;
\item If $\Gamma_1,x:\circ\vdash P\ret \Gamma_2,:x:\circ$ then
$\Gamma_1\vdash P\ret \Gamma_2$;
\item If $\Gamma_1,x:\cotype{\circ}{N}\vdash P\ret \Gamma_2,:
x:\cotype{\circ}{N'}$  then $\Gamma_1,x:N\vdash P\ret \Gamma_2,x:N'$;
\item If $\Gamma_1,x:\cotype{M}{\circ}\vdash P\ret \Gamma_2,
x:\cotype{M'}{\circ}$ then $\Gamma_1,x:M\vdash P\ret \Gamma_2,x:M'$;
\item If $\Gamma_1,x:\un\,p\vdash P\ret \Gamma_2,x:\un\,p$ and $x\not\in\fv(P)$  
then $\Gamma_1 \vdash P\ret \Gamma_2 $;
\item If $\Gamma_1,x:\cotype{\un\,p_1}{\un\,p_2}\vdash P\ret
\Gamma_2,x:\cotype{\un\,p_1}{\un\,p_2}$ and $x\not\in\fv(P)$ 
then $\Gamma_1 \vdash P\ret \Gamma_2 $.
\end{enumerate}
\end{lemma}

\begin{lemma}[Weakening]
\label{lem:weak}
$\I,x:S\vdash_D P \text{ implies }
\I,x:\cotype{S}{\un\,p}\vdash_D P$.
\end{lemma}
 
We have all the ingredients to prove the following result which is the
wedge of the proof of soundness. 
\begin{lemma}\label{lem:soundness}
Assume $\safe(\Gamma_1)$. The following hold.
\begin{enumerate}
\item If $\Gamma_1\vdash v\colon T\ret\Gamma_2$ then  
$\used(\Gamma_1\ret\Gamma_2)\vdash_D v\colon T $;
\item If $\Gamma_1\vdash P\ret\Gamma_2$ then  
$\used(\Gamma_1\ret\Gamma_2)\vdash_D P $.
\end{enumerate}
\end{lemma}
\begin{proof}
We first prove (1).  
Assume  $\Gamma_1,x\colon  \lin\,p ,\Gamma_2\vdash x \colon  \lin\,  p 
    	\ret\Gamma_1,x \colon \circ,\Gamma_2$.
Notice that $\I=\used((\Gamma_1,\Gamma_2)\ret(\Gamma_1,\Gamma_2))$
is a safe type context such that $\un(I)$, i.e. it contains only
unrestricted typings, and that $ \lin\,p \ret\circ=\lin\, p$.
We apply \tvar and infer $\I,x:\lin\,p \vdash_D x:\lin\,p$.
The cases for typing a linear or unrestricted channel type, or an unrestricted
 channel type are analogous.
Assume 
$\Gamma_1,x\colon\cotype{\lin\,p  }{N},\Gamma_2
\vdash
x \colon \lin\, p
    	\ret\Gamma_1,x\colon \cotype{\circ}{N},\Gamma_2
$.
Let $\I=\used((\Gamma_1,\Gamma_2)\ret(\Gamma_1,\Gamma_2))$. 
We have  
$\I(x)=\cotype{\lin\,p}{S}$ with $S=\un\,p'$ or $S=\un\, \End$. 
From these results and \tvar we infer 
$\I,x:\cotype{\lin\,p}{S}\vdash x: \lin\,p   $.
Now assume that
$\Gamma \vdash_D   x \colon \un\, p \ret\Gamma$
with $\Gamma=\Gamma_1,x\colon \cotype{ \un\,p }{N},\Gamma_2$.
From $\un(\used(\Gamma\ret\Gamma))$ and \tvar we infer that there is 
$S=\un\,p'$ or $S= \un\,\End$ such that 
$\un(\used(\Gamma\ret\Gamma))\vdash_D x\colon \cotype{ \un\,p }{S}$.
We apply \tstrength and infer the desired result: 
$ \used(\Gamma\ret\Gamma) \vdash_D x\colon { \un\,p }$.

To prove (2) we proceed by induction on the length of the derivation for
$\Gamma_1\vdash P\ret\Gamma_2$. We prove the most interesting cases.
We use the notation $\Gamma\less x$ to indicate the context $\Gamma'$ whenever
$\Gamma=\Gamma',x:M$ or $\Gamma=\Gamma',x:\cotype{M}{N}$.

\begin{description}
\item{\apar}
We have $ \Gamma_1\vdash P\PAR Q\ret\Gamma_3$ inferred from
$ \Gamma_1\vdash P\ret\Gamma_2$ and 
$ \Gamma_2\vdash Q\ret\Gamma_3$. 
We proceed by case analysis on $\Gamma_1$.
If $\Gamma_1=\emptyset$ we are done by applying the first rule for context
splitting.
Otherwise assume $\Gamma_1=\Gamma,x:T$. We exploit Lemma \ref{lem:usedefined}
in order to infer the type of $\Gamma_2(x)$ and $\Gamma_3(x)$.
\begin{description}
\item{($T=\circ$)}.
We have $\used(\Gamma_1\ret\Gamma_2)(x)=\un\,\End$ and
$\used(\Gamma_2\ret\Gamma_3)(x)=\un\,\End$. We apply \tpar to the I.H by
using the second rule for context splitting to 
$
\used(\Gamma_1\ret\Gamma_2)\vdash_D P
$
and
$
\used(\Gamma_2\ret\Gamma_3)\vdash_D Q
$. 
\item{($T=\cotype{\circ}{\circ}$)}. Analogous to the previous case.
\item{($T=\lin\,p$)}.
We have two cases for
$\used(\Gamma_1\ret\Gamma_2)(x)$ corresponding to 
(i) $\lin\,p$ and (ii) $\un\,\End$. 
In case (i) we have $\used(\Gamma_2\ret\Gamma_3)=\un\,p'$. This is
because by definition of $\used$ we have that $\Gamma_2(x)=\circ$.
By  applying Lemma~\ref{lem:weak} we weaken $\used(\Gamma_1\ret\Gamma_2)$ and
obtain an environment $\Delta$ equal
to $\used(\Gamma_1\ret\Gamma_2)$ but for the entry
$x$ which is weakened to $\cotype{\lin \,p}{\un\,p'}$.
We apply the I.H. and infer the
desired result by applying the fifth rule for context splitting in \tpar:
$
\Delta\vdash_D P
$
and
$
\used(\Gamma_2\ret\Gamma_3)\vdash_D Q
$.
In case (ii) we have $\Gamma_2(x)=\lin p$. In sub case 
$\used(\Gamma_2\ret\Gamma_3)(x)=\lin\,p$ we apply the I.H. and proceed
 by weakening the type to $\cotype{\un\,\End}{\lin\,p}$ in order to apply the
sixth rule for splitting in \tpar.
In sub-case $\used(\Gamma_2\ret\Gamma_3)(x)=\un\,\End$ we apply the second
splitting rule. 
\item{($T=\cotype{ \lin\,p_1 }{ \lin\,p_2 }$)}.
We have four cases for
$\used(\Gamma_1\ret\Gamma_2)(x)$ corresponding to 
(iii) $\cotype{\lin\,p_1}{\lin\,p_2}$ and (iv)
$\cotype{\lin\,p_1}{\un\,\End}$ and (v)  $\cotype{\un\,\End}{\lin\,p_2}$ and
(vi) $\cotype{\un\,\End}{\un\,\End}$.
In case (iii) we infer $\Gamma_2(x)=\cotype{\circ}{\circ}$. We apply
Lemma~\ref{lem:alg-strength} and strengthen the algorithm's judgment by removing
the entry for $x$ in $\Gamma_2$: 
$\Gamma_2\less x\vdash Q\ret\Gamma_3\less x$.
We apply the I.H. and by \tpar we infer the desired result by applying the
fourth rule for context splitting to
$
\used(\Gamma_1\ret\Gamma_2)\vdash_D P
$ and
$
\used(\Gamma_2\less x\ret\Gamma_3\less x)\vdash_D Q \ .
$.
In case (iv) we have $\used(\Gamma_2 \ret\Gamma_3)(x)=\cotype{\un\End}{S}$
where $S=\lin\,p_2$ or $S=\un,\End$. 
If $S=\lin p_2$ we know that $\Gamma_2(x)=\cotype{\circ}{\lin\, p_2}$.
We apply Lemma~\ref{lem:alg-strength} and infer both
$\Gamma,x:\lin\,p_1 \vdash P\ret\Gamma_2\less x, x:\circ$ and
$\Gamma_2\less x,x:\lin p_2\vdash Q\ret \Gamma_3\less x:\circ$.
We apply the I.H. and infer the desired result by applying \tpar with the
fourth rule for context splitting:
$
\used(\Gamma,x:\lin\,p_1\ret\Gamma_2\less x, x:\circ)\vdash_D P
$
and 
$  
\used(\Gamma_2\less x,x:\lin p_2\ret \Gamma_3\less x,x:\circ)\vdash_D Q \ .   
$.
Otherwise when $S=\un\,\End$ by strengthening and I.H. we have
$
\used(\Gamma,x:\lin\,p_1\ret\Gamma_2\less x, x:\circ)\vdash_D P
$ and
$
\used(\Gamma_2\less x\ret \Gamma_3\less x)\vdash_D Q
$ 
and we conclude by applying the third rule for context splitting. 
\item{($T=\cotype{ \lin\,p_1 }{\circ},=\cotype{\circ}{\lin\,p_2}$)}.
Similar to the previous case. 
\item{($T=\un\,p,=\cotype{\un p_1}{\un p_2},=\cotype{\un p_1}{\circ},
=\cotype{\circ}{\un\,p_2}$)}.
The result follows by applying the I.H. and the second rule for context
splitting in \tpar to
$
\used(\Gamma_1\ret\Gamma_2)\vdash_D P
$ and
$\used(\Gamma_2\ret\Gamma_3)\vdash_D Q$. 
\end{description}

\item{\aoutl} We have   $ \Gamma_1,x\colon 
\lin\,\OUT TS \vdash
\sendp\ret\Gamma_3,x:\circ $ inferred from
$\Gamma_1,x:\circ\,\vdash v \colon T\ret\Gamma_2$ 
and $\Gamma_2\uplus  x\colon S \vdash P\ret\Gamma_3,x:M$ 
provided $\Un(M)$. 
By strengthening we infer  
$\Gamma_1\vdash v \colon T\ret\Gamma_2\less x$.
By I.H. we infer
$\used(\Gamma_1,x:  S  \ret\Gamma_3,x:M)\vdash_D P$.
Let $\Delta$ be the environment 
$\used(\Gamma_1,x:  S  \ret\Gamma_3,x:M)$ but such that the type for $x$ in $\Delta$
is equal to $ \lin\,\OUT T(S\ret M)$. 
We apply \tout and the fourth rule for context splitting and we conclude:
$
\used(\Gamma_1\ret\Gamma_2\less x)\csplit\Delta$.  
\item{\aoutcll}
We have    
$\Gamma_1,x\colon \cotype{ \lin\,\OUT TS}{N}
      \vdash \sendp\ret\Gamma_3,x:\cotype{\circ}{N}$
      inferred from
$ \Gamma_1,x\colon \lin\,\OUT TS \vdash \sendp\ret\Gamma_3,x:\circ $. 
By I.H. we infer
$\used( \Gamma_1,x\colon \lin\,\OUT TS\ret\Gamma_3,x:\circ)\vdash_D \sendp$
which we rewrite as
$
\used( \Gamma_1\ret\Gamma_3),x:\lin\,\OUT TS\vdash_D \sendp
$. Since $N\ret N$ is unrestricted, by weakening we infer
$
\used( \Gamma_1\ret\Gamma_3),x:\cotype{\lin\,\OUT TS}{N\ret N}\vdash_D \sendp
$. This is the requested result since 
$\used( \Gamma_1\ret\Gamma_3),x:\cotype{\lin\,\OUT TS}{N\ret N}=
\used( \Gamma_1,x\colon \cotype{\lin\,\OUT
TS}{N})\ret\Gamma_3,x:\cotype{\circ}{N}$.

\item{\arepl}
We have  $\Gamma_1 \vdash  !P\ret\Gamma_2$ inferred from
    $\Gamma_1 \vdash P\ret\Gamma_2$ provided 
    $ \Gamma_1=\Gamma_2$.
By I.H. we have   $\used(\Gamma_1\ret\Gamma_1) \vdash_D P$.
Since $\Gamma_1\ret\Gamma_1$ is an unrestricted context, so is 
$\used(\Gamma_1\ret\Gamma_1)$. 
We apply \trepl and we conclude: $\used(\Gamma_1\ret\Gamma_1) \vdash_D !P$.
\item{\ainact}
We apply \tinact and infer $\used(\Gamma\ret\Gamma)\vdash \INACT$.
\end{description} 
\end{proof}

\vspace{-.5em}
By relying on this result we establish the soundness of the algorithm.
 
\begin{corollary}[Soundness]\label{cor:soundness}
If ${\sf typeCheck}(\I,P)$ then $ \I \vdash_D P$.  
\end{corollary}
\begin{proof}
If the algorithm succeeds then we have 
$\I\vdash P\ret\Gamma$ with $\Un(\Gamma)$. Consider $x\in\dom(\I)$.
If $\I(x)=\lin p$ then we know that $\Gamma(x)=\circ$. Therefore
$\used(\I\ret\Gamma)(x)=\lin p$. Similarly, if $\I(x)=\cotype{\lin p_1}{\lin p_2}$ 
then $\used(\I\ret\Gamma)(x)=\I(x)$. The last possibility is 
$\I(x)=\un p,=\cotype{\un\, p_1}{\un\, p_2}$ and we conclude that 
$\used(\I\ret\Gamma)(x)=\I(x)$. From these facts we infer $\used(\I\ret\Gamma)=\I$.  
The result follows from Lemma \ref{lem:soundness}.
\end{proof}
 
The hypothesis $\safe(\I)$ in ${\sf
typeCheck}(\I,P)$ allows  us to infer that typings are preserved by the system
in Figure \ref{fig:split-system}, in the following sense~\cite{GiuntiV10}. 

\begin{lemma}[Subject reduction]\label{lem:sr}
Assume $\safe(\I)$. If $ \I \vdash_D P$ and $P\Rightarrow P'$ then
$\I'\vdash_D P'$ with $\safe(\I')$. 
\end{lemma}
  
Finally we prove an important result, namely that the algorithm preserves
structural congruence. To tackle the proof, we need a construction similar to
the one of Lemma~\ref{lem:soundness}.    

\begin{lemma}\label{lem:join}
Let $\Gamma_1\vdash P\ret\Gamma_2$. 
We have $\Gamma_1\ret\Gamma_2\vdash P\ret\nabla_{\Gamma_1}$
with   
\[
\nabla_{\Gamma}(x) = 
\left\{
\begin{aligned}
&\circ &\qquad&\Gamma(x)=\lin,p
\\
&\cotype{\circ}{\circ} 
&&\Gamma(x)=\cotype{\lin\,p_1}{\lin\,p_2},
=\cotype{\lin\,p_1}{\circ},=\cotype{\circ}{\lin\,p_2} \\
&\Gamma(x) &&x\in\dom(\Gamma)   
\end{aligned}
\right. 
\]
\end{lemma}
\noindent
Given $\Gamma_1,\Gamma_2$ with the same domain we define the update of contexts
$\Gamma_1,\Gamma_2$ as the operation below:
\[
\Gamma_1\uplus\Gamma_2 =
\left\{
\begin{aligned}
&M_1\uplus M_2 &\qquad&  \Gamma_1(x)=M_1, \Gamma_2(x)=M_2
\\
&\cotype{M_1\uplus N_1}{M_2\uplus N_2} &&  \Gamma_1(x)=\cotype{M_1}{N_1},
\Gamma_2(x)=\cotype{M_2}{N_2}  
\end{aligned}
\right.
\]

\begin{lemma}[Algorithmic weakening]\label{lem:update}
Let $\Gamma_1\vdash P\ret\Gamma_2$. The following hold.
\begin{enumerate}
 \item if $x\not\in\dom(\Gamma)$ then 
 (i)  $\Gamma_1,x:M\vdash P\ret\Gamma_2,x:M $ and 
 (ii) $\Gamma_1,x:\cotype MN\vdash  P\ret\Gamma_2,x:\cotype MN$;
 \item if $\Gamma_1\uplus\Gamma$ is defined then
       $\Gamma_1\uplus\Gamma\vdash P\ret\Gamma_2\uplus\Gamma$.
\end{enumerate}
\end{lemma} 

\begin{lemma}[Structural congruence]\label{lem:congruence}
Assume $P\equiv Q$. We have  
$\Gamma_1\vdash P\ret\Gamma_2$ if and only if 
$\Gamma_1\vdash Q\ret\Gamma_2$.
\end{lemma}
\vspace{-.5em}
\begin{proof}
The most interesting case is parallel composition. 
Assume $\Gamma_1\vdash P\PAR Q\ret\Gamma_3$ inferred from  
$\Gamma_1\vdash P\ret\Gamma_2$ and
$\Gamma_2\vdash Q\ret\Gamma_3$. 
By Lemma~\ref{lem:join} we have
$\Gamma_1\ret\Gamma_2\vdash P\ret\nabla_{\Gamma_1}$ and
$\Gamma_2\ret\Gamma_3\vdash Q\ret\nabla_{\Gamma_2}$. 
In fact, it holds $\nabla_{\Gamma_1}=\nabla=\nabla_{\Gamma_2}$.
Let $\Gamma_4$ be the solution of the linear system defined by  
equations
$\Gamma_1=(\Gamma_2\ret\Gamma_3)\uplus\Gamma_4$ and
$\Gamma_4=(\Gamma_1\ret\Gamma_2)\uplus\Gamma_3$.  
Such a solution does exist (see the Appendix).
By $\Gamma_2\vdash Q\ret\Gamma_3$ and Lemma~\ref{lem:join}
we infer $\Gamma_2\ret\Gamma_3\vdash Q\ret\nabla$. By using 
Lemma~\ref{lem:update} we have
$
\Gamma_2\ret\Gamma_3\uplus\Gamma_4\vdash Q\ret\nabla\uplus\Gamma_4 
$.
Next take $\Gamma_1\ret\Gamma_2\vdash P\ret\nabla$ obtained by applying
Lemma~\ref{lem:join} to
$\Gamma_1\vdash P\ret\Gamma_2$.
We apply Lemma~\ref{lem:update} and infer
$
(\Gamma_1\ret\Gamma_2)\uplus\Gamma_3 \vdash P\ret\nabla\uplus\Gamma_3
$.
Since the update of $\nabla$ with a type environment $\Gamma$, whenever defined, 
satisfies the equation $\nabla\uplus\Gamma=\Gamma$, the judgments above could be
rewritten as $\Gamma_1\vdash Q\ret\Gamma_4$ and 
$\Gamma_4\vdash P\ret\Gamma_3$.
We  apply  \apar and obtain 
$\Gamma_1\vdash Q\PAR P\ret\Gamma_3$, as required. 
The other direction for the parallel case is analogous.
The second rule for congruence of parallel processes is straightforwardly obtained
from the definition of \apar.
The cases for replication and inaction follow easily from the fact that the
context received in output is equal to the context received in input.
The cases for scope restriction follow  from the definition
of \ares and from algorithmic strengthening and weakening
(Lemmas~\ref{lem:alg-strength} and ~\ref{lem:update}). To illustrate, take the rule
$\NR{x:\un\,p}\INACT\equiv \INACT$. Assume $\Gamma\vdash \INACT\ret\Gamma$ 
and let $x\not\in\dom(\Gamma)$, eventually by alpha-renaming $x$ in the left process.
By weakening we infer  $\Gamma, x:\un\,p\vdash \INACT\ret\Gamma, x:\un\,p$. 
We apply \ares and conclude:
$\Gamma \vdash\NR{x:\un\,p}\INACT\ret\Gamma$. Now assume
$\Gamma\vdash \NR{x:\un\,p}\INACT\ret\Gamma_1$ inferred
from $\Gamma,{x:\un\,p}\vdash\INACT\ret\Gamma_1, x:O$. 
From the fact that this judgment has been inferred by using \ainact, we infer
$\Gamma_1=\Gamma$ and $O=\un\,p$. Since  $x\not\in\fv(\INACT)$, by applying
strengthening we infer the desired result, $\Gamma\vdash\INACT\ret \Gamma$.
\end{proof}
  
\begin{theorem}
The ${\sf typeCheck}$ algorithm is effective for establishing a session-based \
type discipline. 
\end{theorem}
\vspace{-.5em}
\begin{proof}
Apply Corollaries~\ref{cor:deterministic},~\ref{cor:soundness} and
Lemmas~\ref{lem:sr} and ~\ref{lem:congruence}. 
\end{proof}

\subsection{Towards semantic completeness}

The algorithm is unable to type check some process that is  
typable by the type system  in Figure~\ref{fig:split-system}. 
This is trivially true for all processes typed by unsafe contexts, 
but also for typings of the form:
\begin{align*}
&\Gamma,x:\cotype{\lin \IN TS}{\lin \OUT T{\dual S}}\vdash_D x().P
&&P\equiv C[\SEND{x}{}P'] 
\\[2mm]
&\Gamma,x:\cotype{\lin \IN TS}{\lin \OUT T{\dual S}}\vdash_D
\SENDn{x}{x} 
&&T=\lin \IN TS\ .
\end{align*}

\noindent
As argued in other works on session types 
(e.g.~\cite{GaySJ:substp,castagna-et-al:foundations-session-types}),  
it seems that ruling out such processes does not comport an issue since they appear
to be deadlocked.
To deploy a formal proof of this statement, we have developed a typed observational 
theory  where the behavior of processes is contrasted w.r.t. the typed knowledge of
the observer~\cite{bisim-sessions}.
The discerning capability of the observer is
regulated by the type checker; in particular, type checking forces  contexts to not
interfere with a session shared by two participants.
Behaviorally equivalent pi calculus processes exhibit the same observables in all
type checked contexts. To avoid universal quantification,  we rely  
on a proof technique based on bisimulation over typed labelled semantics. 
 
The aim is to prove that if $\I\vdash_D P$ has been inferred by using
\tinc or \toutc with a linear channel  type, then $P$ is indistinguishable from 
$\INACT$ in all contexts  type checked by 
a type environment  $Y$ {\it compatible} with $\I$, noted $Y\models P\cong \INACT$.
To illustrate, assume that by applying \tinc we infer  
$\I,x\colon\cotype{\lin\IN TS}{\lin\OUT T{\dual S}}\vdash_D \receivep$.
Intuitively, a process type checked by $Y$ cannot tell apart the input process from 
$\INACT$ because interaction on $x$ is forbidden by $Y$;
the compatibility condition enforces the type environment $Y$ to do not contain 
input or output capabilities of $x$, which are already used in a linear way in $\I$.
Once obtained this result, we should be able to prove our algorithm to be
semantically complete, in the following sense.

\begin{claim}[Completeness]
If $\I_1\vdash_D P_1$ then there are a type environment $\I_2$ and a process $P_2$
s.t. ${\sf typeCheck}(\I_2,P_2)$
and $Y\models P_1\cong P_2$ with $Y$ a type environment compatible with
both $\I_1$ and $\I_2$.
\end{claim}

The idea is to build $P_2$ by descending the derivation tree for $\I_1\vdash P_1$
and by substituting  subtrees of $\I_1\vdash P_1$ with a leaf 
$\I_2\vdash \INACT$ by following two rules:
\begin{description}
 \item{\tinc} $\I_1,x\colon\cotype{\lin\IN TS}{\lin\OUT T\dual S}\vdash
\RECEIVE xyQ$  is exchanged with $\I_2\vdash\INACT$;
 \item{\toutc} $\I_1,x\colon\cotype{\lin\OUT TS}{\lin\IN T\dual S}\vdash
\SEND xvQ$   is exchanged with  $\I_2 \vdash \INACT$; 
\end{description} 
Ideally, we would let $\I_1=\I_2$. Unfortunately, the linear design of the algorithm
forbids this option since the call of the type checking function would return in
output the linear entries not consumed by \ainact. 
This approach indeed works if we relax the linearity of type checking an relies on an
affine setting where each session type is used at most once.
Otherwise, we could prune the linear entries from $\I_1,x\colon\cotype{\lin\IN
TS}{\lin\OUT T\dual S}$ and let the type environment $\I_2$ to contain 
all unrestricted typings in $\I_1$.  The proof is performed by proceeding by
induction while exploiting bisimulation semantics and contextuality of~$\cong$.

As a by-product, this technique could be also useful to detect simple deadlocks
generated by erroneous programming of two opposite linear capabilities in a
sequential way.

\section{Examples}
\label{sec:examples}

The protocol for the scheduling of a meeting discussed in Section
\ref{sec:intro} requires the interaction
 with one or more clients executed in parallel with the service. 
 The bootstrap is due to the interaction with a client process acting as the
 creator of the poll, defined as process $P_2$ below. 
 The process, once it has received the channel for the poll,
sets the title and the date and then sends the invitation for the poll to a
number of recipients by forwarding the channel established to communicate the
date proposals.
 An instance of the protocol is obtained by
 considering the parallel composition of the service $P_1$ and the client $P_2$;
 we let ${\sf string}=\un\,\End= {\sf date}$. 
\begin{align*}
 P_1 &=!x(w).\NR{p:\cotype{S_1}{S_2}}\ (\SEND wp p({\sf title}).p({\sf
date}).!p({\sf date}))
 \\[2mm]
 P_2&=\SEND x{y} y(p).(\SEND p{\sf Meeting}
\SEND p{\sf 17 March}(\SENDn {z_1}{p} \PAR \cdots \PAR\SENDn {z_n}{p} ))
\\[2mm] 
S_1&=\lin\,\IN {\sf string}\lin\,\IN{\sf date}S_3 
\qquad S_3=\un\,\IN{\sf date}S_3
 \\[2mm]
S_2&= \lin\,\OUT {\sf string}\lin\,\OUT{\sf date}S_4 
\qquad S_4 =\un\,\OUT{\sf date}S_4 
\end{align*}

\noindent
By passing the (safe) context $\Gamma$ below to the type checker we obtain that $P_1
\PAR P_2$ is accepted. Notice that, due to Lemma~\ref{lem:congruence},
$P_2 \PAR P_1$ is also accepted; we believe this feature to be of practical
interest. 
For the sake of compactness, in the following we will 
shorten the unrestricted type $\un\,\End$ with $\End$.  
\begin{align*}
&\Gamma=
   x: T_x,  
   y:\cotype{\lin\,\OUT{S_2}\End}{\lin\,\IN{S_2}\End},
   z_1 : \lin\,\OUT{S_4}\End,\dots,z_n:\lin\,\OUT{S_4}\End
 \\[2mm] 
 & T_x=\cotype{\mu a.\un \IN {(\lin\,\OUT{S_2}\End)}a}{\mu b.\un \OUT
              {(\lin\,\OUT{S_2}\End)}b}  
\end{align*}
We present below the most interesting snippets of the execution of 
${\sf typeCheck}(\ \Gamma\, ,\, P_1 \PAR P_2\ )$.

\medskip\noindent
{\bf Typing the (linear) poll delegation}. 
In typing the continuation of $P_1$, the \ares pattern is matched.
Once verified that the type $\cotype{S_1}{S_2}$ is balanced, 
the following sub-call is launched by adding to the context the channel type for
the poll:
\begin{equation}
\Gamma_1={\sf check}(\ \Gamma,w:\lin\,\OUT{S_2}\End,p:\cotype{S_1}{S_2}\ , 
\ \SEND wp p({\sf title}).p({\sf date}).!p({\sf date}) \ )   
\label{ex:ares}
\end{equation}
The call (\ref{ex:ares}) matches the pattern \aoutl and
a call for the continuation is invoked by setting   to void
the sent end point type $S_2$.
\begin{equation}
\Gamma_2={\sf check}(\ \Gamma,w: \End ,p:\cotype{S_1}{\circ}\ , 
\ p({\sf title}).p({\sf date}).!p({\sf date}))  
\label{ex:outl}
\end{equation}

\noindent
When receiving the  context $\Gamma_2$, the pattern \aoutl requires  
$\Gamma_2(w)$ to be unrestricted. The context returned in output to the call in
(\ref{ex:ares}) is obtained by setting $\Gamma_2(w)=\circ$.
When receiving the context $\Gamma_1$, the pattern \ares requires
$\Gamma_1(p)$ to be unrestricted, and the context returned in output to the 
caller is obtained by  removing the entry for $p$ from $\Gamma_1$.

\medskip\noindent
{\bf Typing the replicated receiving of the date}.
In typing the continuation of the  process above the pattern \ainl  is
matched  and the following call is launched by passing as parameter the context 
$\Gamma'=\Gamma,w: \End ,p:  {S_3}  ,{\sf title:string},{\sf date:date}$
:
\vspace{-.5em}
\begin{align}
&\Gamma_3={\sf check}(\ \Gamma'\ \ , \ !p({\sf date}) \ )  
\label{ex:aincll2}	
\end{align}
The pattern \arepl is matched and   the following call is launched.
\begin{align}
&\Gamma_4={\sf check}(\ \Gamma' ,\  p({\sf date})\ )  
\label{ex:ainun} 	
\end{align}

\noindent
To succeed in returning the context in output, \arepl requires the
context  $\Gamma_4$ received from the call~(\ref{ex:ainun}) to be equal to
$\Gamma'$. This is satisfied; in this way we know that any linear resource
has not been used under replication, because that would have implied the presence of
a new void typing.  Finally
the context $\Gamma_3=\Gamma_4$ is returned by \arepl to the caller.

\medskip\noindent
{\bf Typing the (unrestricted) poll delegation}.
In typing the continuation of the client $P_2$, pattern \aoutl is
matched and the following call is launched by passing 
as argument  the context
$\Gamma_3=x:T_x,y:\cotype{\circ}{\End}, p: S_4 ,z_1 :
\lin\,\OUT{S_4}\End,\dots,z_n:\lin\,\OUT{S_4}\End$:
\begin{align}
&\Gamma_5={\sf check}(\Gamma_3\,,\, \SENDn {z_1}{p} \PAR \cdots \PAR\SENDn
{z_n}{p})
\label{ex:apar2}
\end{align} 
The call ({\ref{ex:apar2}) matches the \apar pattern and corresponds
to the forwarding of the poll to the recipients in order to propose a date.
The checking procedure for the first delegation is invoked:  
\begin{align}
&\Gamma_6={\sf check}(\Gamma_3\,,\, \SENDn {z_1}{p})
\label{ex:apar2bis}
\end{align} 
The context $\Gamma_6$ obtained by setting to void the entry for $z_1$ in
$\Gamma_3
$ is returned to the caller (\ref{ex:apar2}) in order to type the next thread.
Lastly context $\Gamma_5$ is  obtained by setting to void the entries for
$z_1,\dots,z_n$ in $\Gamma_3$. 

\begin{remark}
By setting typings to void at the end of
the call for  a linear typing we avoid unsound derivations as the one
below
\[ 
\Gamma_1,x:\lin \OUT {T} \un \IN {T}{S}\stackrel{?}{\vdash} 
\SEND xvP \PAR  \RECEIVE xyQ\ret \Gamma_2,x:\circ \ . 
\] 
 
\noindent
On contrast,  we could type check a standard  use of pi calculus channels by
using the rules for unrestricted channel types of the form 
$T=\cotype{\mu a.\un\,\IN {T'} a}{\mu b.\un\,\OUT {T'} b}$:
\[
\Gamma_1,x:T\vdash \SEND xyP \PAR  \RECEIVE xyQ\ret\Gamma_2,x:T \ .
\]
\end{remark}

\section{Discussion}
\label{sec:discussion}
We have presented a type checking algorithm for establishing a session-based
discipline in (a typed version of) the pi calculus of Milner, Parrow and Walker. 
Following a recent  approach~\cite{vvBeatcs}  
our session types are qualified as linear
or unrestricted; a linear session type could evolve to an unrestricted session
type. Each session type describes one end of the session; the whole session is
described  by a type constructor representing the concurrent behavior of the
two channel ends~\cite{GiuntiV10}.
We assess the soundness of the algorithm by showing that type checked
processes are accepted by a typing system satisfying subject reduction. 

Similarly to other approaches for type checking of linear and session types in the
pi calculus~\cite{kpt:linearity,GaySJ:substp}, we rely on the idea to type a
parallel process $P\PAR Q$ by ignoring the set of linear identifiers used
by $P$ before  type checking~$Q$. By reasoning at the type level, we  
provide for a clean  account of the notion of used identifier by introducing explicit
markers for consumed types. On contrast with the cited approaches, this construction 
let us prove that the algorithm preserves structural congruence, and in turn that
re-arranging of parallel processes is possible; we think that this feature  is of
practical interest.

While the algorithm is not complete, we claim that we are not loosing
expressiveness since the algorithm should type checks all interesting
processes accepted by the split-based typing system. 
We are working on a proof of this result which is based on a typed
observational theory  which permits to contrast the behavior of
processes w.r.t. contexts regulated by type checking~\cite{bisim-sessions}.

Qualified session types are expressive enough to represent  linear
types for lambda
calculus~\cite{walker:substructural-type-systems} and linear and session types
for pi calculus~\cite{kpt:linearity,GaySJ:substp}:
see~\cite{GiuntiV10} for the details.
The presented algorithm is therefore a useful tool to type check systems based
on the notion of linearity of communications.
For instance, the qualified session typing system presented in~\cite{vvsfm09}
for a variant of pi calculus relies on the idea of a double binder to represent
the two ends of a communication. By projecting a qualified session type $S$ into
its dual $\dual S$ we could easily map this construct in our system and in turn
provide a (different) type checking algorithm:
\[
\map {\NR{xy\colon S}P} = \NR{x\colon\cotype{S}{\dual S}}\map{P\subs xy}
\]

It should be noted that the choice of representing computations with a channel
type representing the two ends of the communication rules out some process that
could be interesting. A  process that we are not able to type check is
below.
\begin{align*}
!x(y).\NR{a}(\SEND ya a({\sf title}).a({\sf date}).(!a({\sf date}) \PAR 
\SENDn{a}{{\sf 22March}}) 
\end{align*}
\noindent
The process consists in a modified
version of the poll service where the service itself  proposes a date for
the meeting.  Both the algorithm and the split-based system do  not accept this
process because in the (unrestricted) continuation type  both capabilities would be
needed. While we do not envisage difficulties in introducing subtyping for
unrestricted types  \`a la~\cite{PS96}, this seems to
go in the opposite direction of the idea of channel types. 
We therefore need to investigate  subtyping solutions which take into
account the channel type construct. 

Lastly, a natural completion of this work would be to deploy an algorithm for type
inference. We are convinced that the channel type abstraction  leads to a feasible
implementation based on constraint techniques (e.g.~\cite{mezzina-forte09}).

\begin{sloppypar}
  \paragraph{Acknowledgments.}
The author would like to thank the anonymous referees for  
detailed comments. This work is supported by the ERCIM ABCDE
Programme and by the Comete project, INRIA Saclay-\^Ile de France.
\end{sloppypar}

\bibliography{newsessions}  
\bibliographystyle{eptcs}
 
\appendix
\input{tableaux}

\end{document}

%% file: macros.tex
\newcommand{\typesetproofs}[1]{} 

\newcommand{\I}{I}

\newcommand{\keyword}[1]{\textsf{\upshape\small #1}\xspace}

\newcommand{\PAR}{\mid}
\newcommand{\ret}{\triangleright}

\newcommand{\INACT}{\mathbf 0}

\newcommand{\SENDn}[2]{\overline{#1}\langle{#2}\rangle}
\newcommand{\SEND}[2]{\SENDn{#1}{#2}.}
\newcommand{\RECEIVEn}[2]{{#1}({#2})}
\newcommand{\RECEIVE}[2]{\RECEIVEn{#1}{#2}.}

\newcommand{\NR}[1]{(\nu #1)}

\newcommand{\sendp}{\SEND xyP}
\newcommand{\receivep}{\RECEIVE xyP}

\newcommand{\osred}{\,\rightarrow\,}                
\newcommand{\subs}[2]{[{#1}/{#2}]}

\newcommand{\fv}{\operatorname{fv}}      
\newcommand{\dom}{\operatorname{dom}}

\newcommand{\rulename}[1]{[\text{\sc #1}]\xspace}

\newcommand{\typerulename}[1]{\rulename{T-{#1}}}

\newcommand{\tvar}{\typerulename{Var}}

\newcommand{\tstrength}{\typerulename{Strength}}
\newcommand{\tout}{\typerulename{Out}}
\newcommand{\toutc}{\typerulename{OutC}}

\newcommand{\tin}{\typerulename{In}}
\newcommand{\tinc}{\typerulename{InC}}

\newcommand{\trepl}{\typerulename{Repl}}
\newcommand{\tinact}{\typerulename{Inact}}
\newcommand{\tpar}{\typerulename{Par}}

\newcommand{\tres}{\typerulename{Res}}

 \newcommand{\algrulename}[1]{\rulename{A-{#1}}}
 \newcommand{\avarl}{\algrulename{V-L}}
 \newcommand{\avarlll}{\algrulename{V-LL-l}}
 \newcommand{\avarllr}{\algrulename{V-LL-r}}
 \newcommand{\avarlsr}{\algrulename{V-L-l}}
 \newcommand{\avarlsl}{\algrulename{V-L-r}}
 \newcommand{\avaru}{\algrulename{V-U}}
 \newcommand{\avaruul}{\algrulename{V-UU-l}}
 \newcommand{\avaruur}{\algrulename{V-UU-r}}
 \newcommand{\avarusr}{\algrulename{V-U-l}}
 \newcommand{\avarusl}{\algrulename{V-U-r}}
 \newcommand{\avaruend}{\algrulename{V-EE}}

 \newcommand{\aoutu}{\algrulename{Out-Un}}
 
 \newcommand{\aoutclu}{\algrulename{Out-Un-l}}
 \newcommand{\aoutcru}{\algrulename{Out-Un-r}}
 \newcommand{\aoutl}{\algrulename{Out-L}}
 
 \newcommand{\aoutcll}{\algrulename{Out-L-l}}
 \newcommand{\aoutclr}{\algrulename{Out-L-r}}

 \newcommand{\ainu}{\algrulename{In-Un}}
 
 \newcommand{\ainclu}{\algrulename{In-Un-l}}
  \newcommand{\aincru}{\algrulename{In-Un-r}}
 \newcommand{\ainl}{\algrulename{In-L}}
 
 \newcommand{\aincll}{\algrulename{In-L-l}}
 \newcommand{\ainclr}{\algrulename{In-L-r}}

 \newcommand{\arepl}{\algrulename{
Repl}}
 \newcommand{\ainact}{\algrulename{Inact}}
 \newcommand{\apar}{\algrulename{Par}}
 
 \newcommand{\ares}{\algrulename{Res}}

\newcommand{\reductionrulename}[1]{\rulename{R-{#1}}}

\newcommand{\rcom}{\reductionrulename{Com}}

\newcommand{\rres}{\reductionrulename{Res}}

\newcommand{\rstruct}{\reductionrulename{Struct}}

\newcommand{\rpar}{\reductionrulename{Par}}

\newcommand{\branch}{\&}

\newcommand{\un}{\keyword{un}}
\newcommand{\lin}{\keyword{lin}}

\newcommand{\less}{\backslash}

\newcommand{\typeconst}[1]{\keyword{#1}}

\newcommand{\End}{\typeconst{end}} 

\newcommand{\OUTn}[1]{!{#1}}
\newcommand{\INn}[1]{?{#1}}
\newcommand{\OUT}[1]{\OUTn{#1}.}
\newcommand{\IN}[1]{\INn{#1}.}

\newcommand{\rcdt}{\{l_i\colon T_i\}_{i\in I}}
\newcommand{\brancht}[1][\alpha]{\branch\rcdt}

\newcommand{\recvt}{\IN TS}
\newcommand{\sendt}{\OUT TS}
\newcommand{\rect}{\mu a.S}

\newcommand{\csplit}{\cdot}

\newcommand{\dual}{\overline}
\newcommand{\cotype}[2]{({#1}, {#2})}

\newcommand{\grmeq}{\; ::= \;}

\newcommand{\map}[1]{[\![{#1}]\!]}

\newcommand{\ignore}[1]{}
\newif\ifny\nytrue
\nyfalse

\newif\ifvv\vvtrue
\vvfalse


%% file: fig-syntax.tex
\begin{figure}[t]
  \emph{Types and Processes}
  \begin{align*} 
    q  \grmeq && \text{Qualifiers:} && T \grmeq &&\text{Types:}
    \\
    & \lin & \text{linear} 	&& & 	S &  \text{end point}
    \\
    & \un & \text{unrestricted}&&  & \cotype{S}{S}  & \text{channel}
    \\ 
    p \grmeq & & \text{Pre Types:} && P  \grmeq  & & \text{Processes:}
    \\ 
    & \recvt & \text{receive} &&&  \SEND xyP & \text{output}
    \\
    & \sendt & \text{send}	&&& \RECEIVE xyP & \text{input}
    \\
     &\End &\text{termination} &&& P\PAR P &
\text{composition}
     \\
    S \grmeq & & \text{End Point Types:} &&&
\NR{x:T}P & \text{restriction}
    \\ 
    & q\ p & \text{qualified channel} &&& !P &
\text{replication}
    \\ 
    & a & \text{type variable} &&& \INACT &
\text{inaction}\\
  & \rect & \text{recursive type}  
    \end{align*}
  \\[-1.5em]
  \emph{Rules for structural congruence}
  \begin{gather*}
    P\PAR Q \equiv Q\PAR P
    \qquad
    (P\PAR Q)\PAR R \equiv P\PAR (Q\PAR R)
    \qquad
    P\PAR \INACT\equiv P
    \qquad
    !P \equiv P \PAR !P
    \\
    \NR {x:T} P\PAR Q \equiv \NR {x:T}(P\PAR Q)
    \qquad 
     \NR {x:T_1}\NR {y:T_2}P\equiv \NR {y:T_2}\NR {x:T_1}P
    \\
    \NR {x:\un\,p}\INACT \equiv \INACT
    \qquad
       \NR{x:\cotype{\un\,p_1}{\un\,p_2}}\INACT \equiv \INACT
  \end{gather*}
  \emph{Rules for reduction}  
  \begin{gather*}
    \tag*\rcom
    \SEND x z P \PAR \RECEIVE x y Q  
    \osred
    P \PAR Q\subs z y 
    \\
    \tag*{\rres\rpar\rstruct}
    \frac{
      P \osred  Q
    }{
      \NR{x}P \osred \NR{x}Q
    }
    \qquad
    \frac{
      P \osred  Q
    }{
      P \PAR R \osred Q \PAR R
    }
    \qquad
    \frac{
      P \equiv P' \quad P' \osred Q' \quad Q' \equiv Q
    }{
      P \osred Q
    }
  \end{gather*}
  \caption{Pi calculus}
  \label{fig:syntax}
\end{figure}

%% file: ml-algorithm.tex
\afterpage{\clearpage}
\begin{figure}[!ht]
\begin{footnotesize} 
\setbox0\vbox{\hsize=1\textwidth
\begin{verbatim} 
datatype qualifier = Lin | Un;
datatype preType   = In of sessionType * endpointType | Out of sessionType * endpointType | End
and   endpointType = Qualified of qualifier * preType | Void 
and   sessionType  = EndPoint of endpointType | Channel of endpointType * endpointType; 
type context       = (string * sessionType) list;    
datatype process   = Zero | Replication of process | Parallel of process * process 
                   | Input of string * string * process  | Output of string *  string * process 
                   | New of string * sessionType * process;   
fun safe           (g:context):context;
fun unVar          (g:context,v:string):context;
fun remove         (g:context,v:string):context;
fun setVoid        (g:context,v:string):context;   
fun checkVar       ((x, ((EndPoint (Qualified (Lin,p) ))))::g, 
                    (z, ((EndPoint (Qualified (Lin,r) ))))) =                       
                      if p=r then (x, (EndPoint Void))::g                           (* A-V-L *) 
 |  checkVar       ((x, ((EndPoint (Qualified (Un,p) ))))::g,   
                    (z, ((EndPoint (Qualified (Un,r) ))))) =                         
                      if p=r then ((x, ((EndPoint  ( Qualified (Un,p) ))))::g)      (* A-V-U *)
 |  checkVar       ((x, ((Channel  (Qualified (Lin,p), Qualified (Lin,s) ))))::g,  
                    (z, ((Channel  (Qualified (Lin,r), Qualified (Lin,t) )))))=  
                      if p=r andalso s=t then 
                        ((x, ((Channel  ( Void, Void ))))::g)
                      else
                        if p=t andalso s=r 
                          then ((x, ((Channel  ( Void, Void ))))::g)           (* A-V-LL-L+R *)  
 |  checkVar       ((x, ((Channel  (Qualified (Lin,p), Void ))))::g,   
                    (z, ((EndPoint (Qualified (Lin,r)))))) =    
                      if p=r   
                          then  ((x, ((Channel ( Void,Void ))))::g)               (* A-V-L-R *)
 |  checkVar       ((x, ((Channel  (Qualified (Lin,p), Qualified (Un,s) ))))::g,   
                    (z, ((EndPoint (Qualified (Lin,r)))))) =   
                      if p=r  
                          then  ((x, ((Channel ( Void, Qualified (Un,s)))))::g)   (*A-V-L-R  *) 
(*                | A-V-L-L   | A-V-UU-L  | A-V-U-L  | A-V-U-R    | A-V-E-E                  *); 
fun check          (g:context,Zero:process)=                      
                      g                                                           (* A-INACT *)
 |  check          (g:context,Replication p)=         
                      if g = check(g,p) 
                         then g                                                   (* A-REPL  *)
 |  check          (g:context, Parallel (p1,p2))  =   
                      check ( check (g,p1) , p2 )                                 (* A-PAR   *)                                    
 |  check          ((z,(Endpoint (Qualified (Lin,In (a,c)) )))::t, Input (x,y,p))=       
                      let val d = check ((x,Endpoint (c))::t,p) in 
                        setVoid ( remove (unVar (unVar  (d,x), y), y) , x )
                      end                                                         (* A-IN-L  *)
  | check          (g:context,New (x,t,p) )  =    
                      remove ( unVar ( check ( ( safe([(x,t)]))@g, p) , x ) , x)  (* A-RES   *)
(*                | A-IN-L-l  | A-IN-L-r  | A-IN-Un  | A-IN-Un-l  | A-IN-Un-r | A-OUT-L
                  | A-OUT-L-l | A-OUT-L-r | A-OUT-Un | A-OUT-Un-l | A-OUT-Un-r               *);
fun typeCheck      (g:context,p:process)  =   
                      un ( check ( safe(g) , p ) );                         
\end{verbatim}
} 
\fbox{\box0}
\end{footnotesize}
\caption{ML code of the algorithm (excerpt)}
\label{fig:algorithm}
\end{figure}

 

%% file: fig-type-system.tex
\begin{figure}[t]
 \emph{Context splitting rules}
  \begin{gather*}
    \emptyset = \emptyset \csplit \emptyset
    \qquad \qquad    
    \frac{
      \I = \I_1 \csplit \I_2
      \qquad
      T = \un\, p \text{ or } \cotype{\un\, p_1}{\un\, p_2}
    }{
      \I, x\colon T = (\I_1,x\colon T) \csplit (\I_2,x\colon T)
    }
    \\
    \frac{
      \I = \I_1 \csplit \I_2
      \qquad
      T = \lin\, p \text{ or } \cotype{\lin\, p_1}{\lin\, p_2}
    }{
      \I, x\colon T = (\I_1,x\colon T) \csplit \I_2      
    }
    \qquad
    \frac{
      \I = \I_1 \csplit \I_2
      \qquad
      T = \lin\, p \text{ or } \cotype{\lin\, p_1}{\lin\, p_2} 
    }{
      \I, x\colon T = \I_1 \csplit (\I_2,x\colon T)
    }
    \\
    \frac{
      \I = \I_1 \csplit \I_2
    }{
      \I, x\colon\cotype{\lin\, p_1}{\lin\, p_2} = (\I_1,x\colon \lin\, p_1)
      \csplit (\I_2,x\colon \lin\, p_2)
    }
  \\
    \frac{
      \I = \I_1 \csplit \I_2
    }{
      \I, x\colon\cotype{\lin\, p_1}{\un\, p_2} = (\I_1,x\colon
\cotype{\lin\, p_1}{\un\, p_2})
      \csplit (\I_2,x\colon \un\, p_2)
    }
    \\
    \frac{
      \I = \I_1 \csplit \I_2
    }{
      \I, x\colon\cotype{\lin\, p_1}{\un\, p_2} = (\I_1,x\colon \un\, p_2)
      \csplit (\I_2,x\colon\cotype{\lin\, p_1}{\un\, p_2})
    }
 \end{gather*}
  \emph{Typing rules for values}
  \begin{gather*}
    \tag*{\tvar\tstrength}
    \frac{
      \Un(\I)
    }{
      \I,x\colon T\vdash_D x \colon T
    }
    \qquad
    \frac{
      \I \vdash_D v \colon \cotype{S}{\un\,p} 
    }{
      \I \vdash_D v \colon S
    }
  \end{gather*}
  \emph{Typing rules for processes}
  \begin{gather*}
    \tag*{\tinact\tpar}
    \frac{
      \Un(\I)
    }{
      \I \vdash_D \INACT
    }      
    \qquad
    \frac{
      \I_1 \vdash_D R_1
      \qquad
      \I_2 \vdash_D R_2
    }{
      \I_1 \csplit \I_2 \vdash_D R_1 \PAR R_2
    }
   \\
    \tag*{\trepl\tres}
        \frac{
      \I \vdash_D R
      \qquad
      \un(\I)
    }{
      \I \vdash_D !R
    }
    \qquad
    \frac{
      \I_,x\colon T \vdash_D R
      \qquad
      \safe(T) 
    }
    {
      \I \vdash_D \NR{x\colon T}R
    }
    \\
    \tag*{\tin,\tout}
    \frac{
      \I,x\colon S,y\colon T   \vdash_D R
      \qquad (*)
    }{
      \I,x\colon q\IN TS\vdash_D \receivep
    }
    \qquad
    \frac{
      \I_1\vdash_D v \colon T
      \qquad
      \I_2, x\colon S\vdash_D R
      \qquad (**)
    }{
      \I_1\csplit(\I_2, x\colon q\,\OUT TS)\vdash_D \sendp
    }
    \\
    \tag*{\tinc,\toutc}
    \frac{ 
      \I,x\colon \cotype{S}{S'},y\colon T\vdash_D R
      \qquad (*)
    }{
      \I,x\colon\cotype{q\IN TS}{S'}\vdash_D \receivep
    }
    \qquad
    \frac{
      \I_1\vdash_D v \colon T
      \qquad
      \I_2, x\colon \cotype{S}{S'}\vdash_D R
      \qquad (**)
    }{
      \I_1\csplit(\I_2, x\colon \cotype{q\,\OUT TS}{S'})\vdash_D \sendp
    }
    \\
    (*)\ q=\un \Rightarrow q\IN TS=S
    \qquad\qquad
    (**)\ q=\un \Rightarrow q\OUT TS=S
  \end{gather*}
  \caption{Split-based typing system}
  \label{fig:split-system}
\end{figure}


%% file: tableaux.tex
\section{Appendix}
The table in Figure~\ref{fig:tableaux} depicts the shape of contexts used in 
the proof of the case of congruence of parallel processes in 
Lemma~\ref{lem:congruence}.
The first three columns in the table represent all
possible combinations for (an entry of) safe contexts $\Gamma_1,\Gamma_2$ and
$\Gamma_3$ such that
\[ 
\Gamma_1\vdash P\ret\Gamma_2 
\ \text{ and } \ 
\Gamma_2\vdash Q\ret\Gamma_3 
\]
Given these inputs, the next three columns show the output for the
context in the header. 
Context $\Gamma_4$ in the seventh column is the solution  of the following linear
system:
\[
\left\{ 
\begin{aligned}
\Gamma_1&=(\Gamma_2\ret\Gamma_3)\uplus\Gamma_4
\\
\Gamma_4&=(\Gamma_1\ret\Gamma_2)\uplus\Gamma_3 
\end{aligned}
\right.
\]
\noindent
In the last column 
we have the environment $\nabla_{\Gamma_1}=\nabla=\nabla_{\Gamma_2}$.
\afterpage{\clearpage}
\begin{figure}[!ht]
\begin{small}
\[
\begin{array}{|c|c|c|c|c|c|c|c|c|}
\hline
\Gamma_1 &\Gamma_2 &\Gamma_3  &\Gamma_1\ret\Gamma_2    & \Gamma_2\ret\Gamma_3
&\Gamma_1\ret\Gamma_3 &\Gamma_4  &\nabla  
\\
\hline 
\lin\,p &\lin\,p &\lin\,p &\circ   & \circ &\circ & \lin\,p &
\circ \\[2mm]
\lin\,p &\lin\,p & \circ &\circ &\lin\,p &\lin\,p  &\circ &\circ
\\[2mm]
\lin\,p &\circ &\circ &\lin\,p &\circ &\circ &\lin\,p &\circ 
\\[2mm]
\un\,p &\un\,p &\un\,p &\un\,p &\un\,p &\un\,p &\un\,p &\un\,p
\\[2mm]
\circ &\circ &\circ &\circ &\circ &\circ &\circ  &\circ
\\[2mm]
\cotype{\lin p_1}{\lin p_2} & \cotype{\lin p_1}{\lin p_2}
&\cotype{\lin p_1}{\lin p_2} &\cotype{\circ}{\circ} 
&\cotype{\circ}{\circ} &\cotype{\circ}{\circ}  
&\cotype{\lin p_1}{\lin p_2}
&\cotype{\circ}{\circ} 
\\[2mm]
\cotype{\lin p_1}{\lin p_2} & 
\cotype{\lin p_1}{\lin p_2} &
\cotype{\lin p_1}{\circ} &
\cotype{\circ}{\circ} &
\cotype{\circ}{\lin p_2} & 
\cotype{\circ}{\lin p_2} & 
\cotype{\lin p_1}{\circ} &
\cotype{\circ}{\circ}  
\\[2mm]
\cotype{\lin p_1}{\lin p_2} & 
\cotype{\lin p_1}{\circ} &
\cotype{\lin p_1}{\circ} &
\cotype{\circ}{\lin p_2} &
\cotype{\circ}{\circ} & 
\cotype{\circ}{\lin p_2} & 
\cotype{\lin p_1}{\lin p_2} &
\cotype{\circ}{\circ}  
\\[2mm] 
\cotype{\lin p_1}{\lin p_2} & 
\cotype{\lin p_1}{\lin p_2} &
\cotype{\circ}{\lin p_2} &
\cotype{\circ}{\circ} &
\cotype{\lin p_1}{\circ} & 
\cotype{\lin p_1}{\circ} & 
\cotype{\circ}{\lin p_2} &
\cotype{\circ}{\circ}  
\\[2mm]
\cotype{\lin p_1}{\lin p_2} & 
\cotype{\circ}{\lin p_2} &
\cotype{\circ}{\lin p_2} &
\cotype{\lin p_1}{\circ} &
\cotype{\circ}{\circ} & 
\cotype{\lin p_1}{\circ} & 
\cotype{\lin p_1}{\lin p_2} &
\cotype{\circ}{\circ}  
\\[2mm]
\cotype{\lin p_1}{\lin p_2} & 
\cotype{\lin p_1}{\lin p_2} &
\cotype{\circ}{\circ} &
\cotype{\circ}{\circ} &
\cotype{\lin p_1}{\lin p_2} &
\cotype{\lin p_1}{\lin p_2} &
\cotype{\circ}{\circ} &
\cotype{\circ}{\circ}
\\[2mm]
\cotype{\lin p_1}{\lin p_2} & 
\cotype{\lin p_1}{\circ} &
\cotype{\circ}{\circ} &
\cotype{\circ}{\lin p_2} &
\cotype{\lin p_1}{\circ} &
\cotype{\lin p_1}{\lin p_2} &
\cotype{\circ}{\lin p_2} &
\cotype{\circ}{\circ}
\\[2mm]
\cotype{\lin p_1}{\lin p_2} & 
\cotype{\circ}{\lin p_2} &
\cotype{\circ}{\circ} &
\cotype{\lin p_1}{\circ} &
\cotype{\circ}{\lin p_2} &
\cotype{\lin p_1}{\lin p_2} &
\cotype{\lin p_1}{\circ} &
\cotype{\circ}{\circ}
\\[2mm]
\cotype{\lin p_1}{\lin p_2} & 
\cotype{\circ}{\circ} &
\cotype{\circ}{\circ} &
\cotype{\lin p_1}{\lin p_2} &
\cotype{\circ}{\circ} &
\cotype{\lin p_1}{\lin p_2} &
\cotype{\lin p_1}{\lin p_2} &
\cotype{\circ}{\circ}
\\[2mm]
\cotype{\lin p_1}{\circ} & 
\cotype{\lin p_1}{\circ} &
\cotype{\lin p_1}{\circ} &
\cotype{\circ}{\circ} &
\cotype{\circ}{\circ} &
\cotype{\circ}{\circ} &
\cotype{\lin p_1}{\circ} &
\cotype{\circ}{\circ}
\\[2mm]
\cotype{\lin p_1}{\circ} & 
\cotype{\lin p_1}{\circ} &
\cotype{\circ}{\circ} &
\cotype{\circ}{\circ} &
\cotype{\lin p_1}{\circ} &
\cotype{\lin p_1}{\circ} &
\cotype{\circ}{\circ} &
\cotype{\circ}{\circ}
\\[2mm]
\cotype{\lin p_1}{\circ} & 
\cotype{\circ}{\circ} &
\cotype{\circ}{\circ} &
\cotype{\lin p_1}{\circ} &
\cotype{\circ}{\circ} &
\cotype{\lin p_1}{\circ} &
\cotype{\lin p_1}{\circ} &
\cotype{\circ}{\circ} 
\\[2mm]
\cotype{\circ}{\lin p_1} & 
\cotype{\circ}{\lin p_1} &
\cotype{\circ}{\lin p_1} &
\cotype{\circ}{\circ} &
\cotype{\circ}{\circ} &
\cotype{\circ}{\circ} &
\cotype{\circ}{\lin p_1} &
\cotype{\circ}{\circ}
\\[2mm]
\cotype{\circ}{\lin p_1} & 
\cotype{\circ}{\lin p_1} &
\cotype{\circ}{\circ} &
\cotype{\circ}{\circ} &
\cotype{\circ}{\lin p_1} &
\cotype{\circ}{\lin p_1} &
\cotype{\circ}{\circ} &
\cotype{\circ}{\circ}
\\[2mm]
\cotype{\circ}{\lin p_1} & 
\cotype{\circ}{\circ} &
\cotype{\circ}{\circ} &
\cotype{\circ}{\lin p_1} &
\cotype{\circ}{\circ} &
\cotype{\circ}{\lin p_1} &
\cotype{\circ}{\lin p_1} &
\cotype{\circ}{\circ}
\\[2mm]
\cotype{\un p_1}{\un p_2} &
\cotype{\un p_1}{\un p_2} &
\cotype{\un p_1}{\un p_2} &
\cotype{\un p_1}{\un p_2} &
\cotype{\un p_1}{\un p_2} &
\cotype{\un p_1}{\un p_2} &
\cotype{\un p_1}{\un p_2} &
\cotype{\un p_1}{\un p_2}
\\[2mm]
\cotype{\un p_1}{\circ} &
\cotype{\un p_1}{\circ} &
\cotype{\un p_1}{\circ} &
\cotype{\un p_1}{\circ} &
\cotype{\un p_1}{\circ} &
\cotype{\un p_1}{\circ} &
\cotype{\un p_1}{\circ} &
\cotype{\un p_1}{\circ}  
\\[2mm]
\cotype{\circ}{\un p_2} &
\cotype{\circ}{\un p_2} &
\cotype{\circ}{\un p_2} &
\cotype{\circ}{\un p_2} &
\cotype{\circ}{\un p_2} &
\cotype{\circ}{\un p_2} &
\cotype{\circ}{\un p_2} &
\cotype{\circ}{\un p_2}  
\\[2mm]
\cotype{\circ}{\circ} &
\cotype{\circ}{\circ} &
\cotype{\circ}{\circ} &
\cotype{\circ}{\circ} &
\cotype{\circ}{\circ} &
\cotype{\circ}{\circ} &
\cotype{\circ}{\circ} &
\cotype{\circ}{\circ}
\\ 
\hline
\end{array}
\]
\end{small}
\caption{Contexts in \apar}
\label{fig:tableaux}
\end{figure}